%% file: main.tex
\documentclass[11pt,letterpaper]{article}

\usepackage{mystyle} 

\newcommand{\SYS}{\text{SYS}}
\newcommand{\USER}{\text{USER}}
\newcommand{\NET}{\text{NET}}
\newcommand{\SYSR}{\text{SYS\_REL}}
\newcommand{\SYSA}{\text{SYS\_AVG}}
\newcommand{\USERA}{\text{USER\_AVG}}
\newcommand{\SYSRES}{\text{SYS\_FIX}}

\title{Optimal Resource Allocation over Networks 
via Lottery-Based Mechanisms}
\author{Soham R.\ Phade and Venkat Anantharam
\thanks{Research supported by the NSF Science and Technology Center grant CCF-
0939370: "Science of Information", the NSF grants ECCS-1343398, CNS-1527846
and CIF-1618145, and the William and Flora Hewlett Foundation supported Center for Long Term Cybersecurity at Berkeley.}
\thanks{The authors are with the Department of Electrical Engineering and Computer Science, University of California, Berkeley, Berkeley, CA 94720.
        {\tt\small soham\_phade@berkeley.edu, ananth@eecs.berkeley.edu}}%
}
\date{}

\begin{document}
\maketitle

\begin{abstract}

 We show that, in a resource allocation problem,
 the ex ante aggregate utility of players with cumulative-prospect-theoretic preferences can be increased over deterministic allocations by implementing lotteries. 
 We formulate an optimization problem, called the system problem, to find the optimal lottery allocation. 
 The system problem exhibits a two-layer structure comprised of a permutation profile and optimal allocations given the permutation profile.
 For any fixed permutation profile, we provide a market-based mechanism to find the optimal allocations and prove the existence of equilibrium prices. 
 We show that the system problem has a duality gap, in general, and that the primal problem is NP-hard. 
 We then consider a relaxation of the system problem and derive some qualitative features of the optimal lottery structure.

\end{abstract}

\input{./tex/intro.tex}

\input{tex/discrete_prob}

\input{tex/equilibrium_decomp}
\input{tex/opt_pi}

\input{tex/relaxation}

\input{tex/conclusion}

\appendix
\input{tex/appendix}

\bibliographystyle{abbrv} 
\bibliography{Bib_Database}

\end{document}

%% file: tex/intro.tex

\section{Introduction}

We consider the problem of congestion management in a network, 
and resource allocation amongst heterogeneous users, in particular human agents, with varying preferences. 
This is a well-recognized problem in network economics, 
with applications to transportation and
telecommunication networks, 
energy smart grids, 
information and financial networks, 
labor markets and social networks, to name a few \cite{nagurney2013network}. 
Market-based solutions have proven to be very useful for this purpose, 
with varied mechanisms, such as auctions 
and fixed rate pricing \cite{falkner2000overview}.
In this paper, we consider a lottery-based mechanism, as opposed to the deterministic allocations studied in the literature. 
We mainly ask the following questions:
\begin{inparaenum}[(i)]
	\item \textit{Do lotteries provide an advantage over deterministic implementations?}
	\item \textit{If yes, then does there exist a market-based mechanism to implement an optimum lottery?}
\end{inparaenum}

In order to answer the first question we need to define our goal in allocating resources. 
There is an extensive literature on the advantages of lotteries: Eckhoff \cite{eckhoff1989lotteries} and Stone \cite{stone2007lotteries} hold that lotteries are used because of fairness concerns;
Boyce \cite{boyce1994allocation} argues that lotteries are effective to reduce rent-seeking from speculators;
Morgan \cite{morgan2000financing} shows that lotteries are an effective way of financing public goods through voluntary funds, when the entity raising funds lacks tax power;
Hylland and Zeckhauser \cite{hylland1979efficient} propose implementing lotteries to elicit honest preferences and allocate jobs efficiently.
In all of these works, there is an underlying assumption, which is also one of the key reasons for the use of lotteries, that the goods to be allocated are indivisible.

However, we notice lotteries being implemented even when the goods to be allocated are divisible, for example in lottos and parimutuel betting. 
In several experiments \cite{prabhakar2013designing}, it has been  observed that lottery-based rewards are more appealing than deterministic rewards of the same expected value, 
and thus provide an advantage in maximizing the desired influence on people's behavior.
We also observe several firms presenting lottery-based offers to incentivize customers into buying their products or using their services, and in return to improve their revenues.
Thus, although lottery-based mechanisms are being widely implemented, a theoretical understanding for the same seems to be lacking. 
This is one of the motivations for this paper, which aims to justify the use of lottery-based mechanisms, 
based on models coming from behavioral economics for how humans evaluate options.


We take a utilitarian approach of maximizing the ex ante aggregate utility or the net happiness of the players. 
(See~\cite{altman2004equilibrium} and the references therein for the relation with other goals such as maximizing revenue.) 
We model each player's utility using cumulative prospect theory (CPT), 
a framework pioneered by Tversky and Kahneman \cite{tversky1992advances}, 
which is believed, based on extensive experimentation with human subjects \cite{camerer1998prospect}, to form a better theory 
with which to model human behavior when faced with prospects
 than is expected utility theory (EUT) \cite{von1945theory}. 
It is important to emphasize that CPT includes EUT as a special case and therefore 
provides a strict generalization of existing modeling techniques.

CPT posits a \emph{probability weighting function} that, along with the ordering of the allocation outcomes in a lottery,
dictates the \emph{probabilistic sensitivity} of a player (details in Section~\ref{sec: model}),
a property that plays an important role in lotteries and gambling. 
As Boyce \cite{boyce1994allocation} points out, ``it is the lure of getting the good without having to pay for it that gives allocation by lottery its appeal.'' 
The probability weighting function typically over-weights small probabilities and under-weights large probabilities, and this captures the ``lure'' effect.  
CPT also posits a reference point that divides the prospect outcomes into gains and losses domains in order to model the loss aversion of the players. 
In order to focus on the effects of probabilistic sensitivity, 
and to avoid the complications resulting from reference point considerations, 
we assume that 
the reference point of all the players is equal to $0$, and we consider prospects with only nonnegative outcomes. 
This is, in fact, identical to the rank dependent utility (RDU) model \cite{quiggin1982theory}.


We will mainly be concerned with the framework considered in \cite{kelly1997charging}, 
that of throughput control in the internet with elastic traffic. 
However, this framework is general enough to have applications to network resource allocation problems arising in several other domains.
Kelly suggested that the throughput allocation problem be posed as one of achieving maximum aggregate utility for the users. 
A market is proposed, in which each user submits an amount she is willing to pay per unit time to the network 
based on tentative rates that she received from the network; 
the network accepts these submitted amounts and determines the price of each network link. 
A user is then allocated a throughput in proportion to her submitted amount and inversely proportional to the sum of the prices of the links she wishes to use. 
Under certain assumptions, Kelly shows that there exist equilibrium prices and throughput allocations, 
and that these allocations achieve maximum aggregate utility. 
Thus the overall \emph{system problem} of maximizing aggregate utility is decomposed into a \emph{network problem} and several \emph{user problems}, 
one for each individual user.
Further, in \cite{kelly1998rate}, the authors have proposed two classes of algorithms which can be used to implement a relaxation of the above optimization problem.

Instead of allocating a single throughput, we consider allocating a \emph{prospect} of throughputs to each user. 
Such a prospect consists of a finite set of throughputs and a probability assigned to each of these throughputs, 
with the interpretation that one of these throughputs would be realized with its corresponding probability (see Section \ref{sec: model} for the definition).
We then ask the question of finding the optimum allocation profile of prospects,
one for each user, comprised of throughputs and 
associated probabilities for that user, that maximizes the aggregate utility of all the players, and is also \emph{feasible}. 
An allocation profile of prospects for each user is said to be feasible if it can be implemented, i.e. there exists a probability distribution over feasible throughput allocations whose marginals for each player agree with their allocated prospects.


If all the players have EUT utility with concave utility function, as is typically assumed to model risk-averseness,  
one can show that there exists a feasible deterministic allocation that achieves the optimum and hence there is no need to consider lotteries. 
However, if the players' utility is modeled by CPT, then one can improve over the best aggregate utility obtained through deterministic allocations.

For example, Quiggin \cite{quiggin1991optimal} considers the problem of distributing a fixed amount amongst several homogeneous players with RDU preferences. 
He concludes that, under certain conditions on players' RDU preferences, the optimum allocation system is a lottery scheme with a few large prizes and a large number of small prizes,
and is strictly preferred over distributing the total amount deterministically amongst the players. 
In Section~\ref{sec: relaxtion}, we extend these results to network settings with heterogeneous players.

In Section \ref{sec: model}, we describe the network model, the lottery structure, and the CPT model of player utility.
We formulate an optimization problem, called the system problem,
to find the optimum lottery scheme.
The solution of such an optimization problem, 
as explained in Section \ref{sec: model}, 
exhibits a layered structure of finding a \emph{permutation profile}, 
and corresponding feasible throughput allocations.
The permutation profile dictates in which order throughput allocations for each player are coupled together for network feasibility purposes. 
A player's CPT value for her lottery allocation depends only on the order amongst her own throughputs,
and not on the coupling with the other players.

Given a permutation profile, the problem of finding optimum feasible throughput allocations is a convex programming problem, which we call the fixed-permutation system problem, and leads to a nice price mechanism. 
In Section \ref{sec: equilibrium}, we prove the existence of equilibrium prices that decompose the fixed-permutation system problem into a network problem and several user problems, one for each player, as in \cite{kelly1997charging}. 
The prices can be interpreted as the cost imposed on the players and can be implemented in several forms, such as waiting times in waiting-line auctions or first-come-first-served allocations \cite{barzel1974theory,taylor2003lottery}, 
delay or packet loss in the Internet TCP protocol \cite{mo2000fair,la2002utility}, 
efforts or resources invested by players in a contest \cite{moldovanu2001optimal,che2003optimal}, 
or simply money or reward points.

Finding the optimum permutation profile, on the other hand, is a non-convex problem. 
In Section~\ref{sec: permutation_duality_gap}, we study the duality gap in the system problem and consider a relaxation of the system problem 
by allowing the permutations to be doubly stochastic matrices 
instead of restricting them to be permutation matrices. 
We show that strong duality holds in the relaxed system problem and so it has value equal to the dual of the original system problem (Theorem~\ref{thm: duality_ineq}). 
We also consider the problem where link constraints hold in expectation, called the average system problem, and show that strong duality holds in this case and so it has value equal to the relaxed problem.
In Section~\ref{sec: relaxtion}, we study the average system problem in further detail, and prove a result on the structure of optimal lotteries. 
Example~\ref{ex: duality_gap}, establishes that the duality gap in the original system problem can be nonzero and Theorem~\ref{thm: nphardness} shows that the primal system problem is NP-hard.
In Section \ref{sec: conclusion}, we conclude with some open problems for future research.

%% file: tex/discrete_prob.tex

\section{Model}
\label{sec: model}

Consider a network with a set $[m] = \{1,\dots,m\}$ of \emph{resources} or \emph{links} 
and a set $[n] = \{1,\dots,n\}$ of \emph{users} or \emph{players}. 
Let $c_j > 0$ denote the finite \emph{capacity} of link $j \in [m]$ and 
let $c := (c_j)_{j \in [m]} \in \bbR^m$. 
(All vectors, unless otherwise specified, will be treated as column vectors.)
Each user $i$ has a fixed \emph{route} $J_i$, 
which is a non-empty subset of $[m]$. 
Let $A$ be an $n \times m$ matrix, 
where $A_{ij} = 1$ if link $j \in J_i$, and $A_{ij} = 0$ otherwise. 
Let $x := (x_i)_{i \in [n]} \in \bbR^n_+$ denote an \emph{allocation profile} 
where user $i$ is allocated the throughput $x_i \geq 0$ that flows through the links in the route $J_i$. 
We say that an allocation profile $x$ is feasible if it satisfies the capacity constraints of the network, i.e., $A^T x \leq c$,
where the inequality is coordinatewise.
Let $\pcal F$ denote the set of all feasible allocation profiles.
We assume that the network constraints are such that $\pcal {F}$ is bounded, and hence a polytope.

Instead of allocating a fixed throughput $x_i$ to player $i \in [n]$, we consider allocating her a \emph{lottery} (or a \emph{prospect})
\begin{equation}
\label{eq: lottery_struct}
	L_i := \{(p_i(1),y_i(1)),\dots,(p_i(k_i),y_i(k_i))\},
\end{equation}
where $y_i(l_i) \geq 0, l_i \in [k_i]$, denotes a throughput and $p_i(l_i), l_i \in [k_i]$, is the probability with which throughput $y_i(l_i)$ is allocated. 
We assume the lottery to be exhaustive, i.e. $\sum_{l_i \in [k_i]} p_i(l_i) = 1$. 
(Note that we are allowed to have $p_i(l_i) = 0$ for some values of $l_i \in [k_i]$ and $y_i(l_i^1) = y_i(l_i^2)$ for some $l_i^1,l_i^2 \in [k_i]$.)
Let $L = (L_i, i \in [n])$ denote a \emph{lottery profile}, where each player $i$ is allocated lottery $L_i$.

We now describe the CPT model we use to measure the ``utility'' or ``happiness'' derived by each player from her lottery (for more details see~\cite{wakker2010prospect}). 
Each player $i$ is associated with a value function
$v_i : \bbR_+ \to \bbR_+$ 
that is continuous, differentiable, concave, and strictly increasing, 
and a probability weighting function $w_i : [0,1] \to [0,1]$ that is continuous, 
strictly increasing and satisfies $w_i(0) = 0$ and $w_i(1) = 1$.

For the prospect $L_i$ in \eqref{eq: lottery_struct}, let $\pi_i : [k_i] \to [k_i]$ be a permutation such that
\begin{equation}
	\label{eq: pi_order_1}
	z_i(1) \geq z_i(2) \geq  \dots \geq z_i(k_i),
\end{equation}
and
\begin{equation}
	\label{eq: pi_order_2}
	y_i(l_i) = z_i(\pi_i(l_i)) \mbox{ for all } l_i \in [k_i].
\end{equation}
The prospect $L_i$ can equivalently be written as
\[
	L_i = \{(\tilde p_i(1), z_i(1)); \dots ; (\tilde p_i(k_i), z_i(k_i))\},
\]
where $\tilde p_i(l_i) := p_i(\pi_i^{-1}(l_i))$ for all $l_i \in [k_i]$.
The \emph{CPT value} of prospect $L_i$ for player $i$ is evaluated using the value function $v_i(\cdot)$ and the probability weighting function $w_i(\cdot)$ as follows:
\begin{equation}
\label{eq: RDU_formula}
	V_i(L_i) := \sum_{l_i = 1}^{k_i} d_{l_i}(p_i,\pi_i)v_i(z_i(l_i)),
\end{equation}
where $d_{l_i}(p_i,\pi_i)$ are the \emph{decision weights} given by $d_1(p_i,\pi_i) := w_i(\tilde p_i(1))$ and
\begin{align*}
	d_{l_i}(p_i,\pi_i) := w_i(\tilde p_i(1) + \dots + \tilde p_i(l_i)) - w_i(\tilde p_i(1) + \dots + \tilde p_i(l_i - 1)),
\end{align*}
for $1 < l_i \leq k_i$.
Although the expression on the right in equation \eqref{eq: RDU_formula} depends on the permutation $\pi_i$, one can check that the formula evaluates to the same value $V_i(L_i)$ as long as $\pi_i$ satisfies \eqref{eq: pi_order_1} and \eqref{eq: pi_order_2}. 
The CPT value of prospect $L_i$, can equivalently be written as
\[
	V_i(L_i) = \sum_{l_i = 1}^{k_i} w_i\big(\sum_{s_i = 1}^{l_i} \tilde p_i(s_i)\big ) \l[ v_i(z_i(l_i)) - v_i(z_i(l_i + 1))) \r],
\]
where $z_i(k_i+1) := 0$.
Thus the lowest allocation $z_i(k_i)$ is weighted by $w_i(1) = 1$, and every increment in the value of the allocations, $v_i(z_i(l_i)) - v_i(z_i(l_i + 1)), \forall l_i \in [k_i -1]$, is weighted by the probability weighting function of the probability of receiving an allocation at least equal to $z_i(l_i)$.

For any finite set $S$, let $\Delta(S)$ denote the standard simplex of all probability distributions on the set $S$, i.e.,
\[
	\Delta(S) := \{(p(s), s \in S)| p(s) \geq 0 \; \forall s \in S, \sum_{s \in S} p(s) = 1\}.
\]
Thus $p_i := (p_i(l_i))_{l_i \in [k_i]} \in \Delta([k_i])$.
We say that a lottery profile $L$ is \emph{feasible} 
if there exists a joint distribution $p \in \Delta(\prod_i [k_i])$ 
such that the following conditions are satisfied:
\begin{enumerate}[(i)]
	\item The marginal distributions agree with $L_i$ for all players $i$, i.e. $\sum_{l_{-i}} p(l_i,l_{-i}) = p_i(l_i)$ for all $l_i \in [k_i]$, where $l_{-i}$ in the summation ranges over values in $\prod_{i' \neq i} [k_{i'}]$.
	\item For each $(l_i)_{i \in [n]} \in \prod_i [k_i]$ in the support of the distribution $p$ 
	(i.e. $p((l_i)_{i \in [n]}) > 0$), the allocation profile $(y_i(l_i))_{i \in [n]}$ is feasible.
\end{enumerate}

The distribution $p$ and the throughputs $\l(y_i(l_i), i \in [n], l_i \in [k_i]\r)$ of a feasible lottery profile together define a \emph{lottery scheme}. 
In the following, we restrict our attention to specific types of lottery schemes, wherein the network implements with equal probability 
one of the $k$ allocation profiles
$y(l) := (y_i(l))_{i \in [n]} \in \bbR^n_+, \text{ for } l \in [k]$.
Let $[k] = \{1, \dots, k\}$ denote the set of \emph{outcomes}, where allocation profile $y_i(l)$
 is implemented if outcome $l$ occurs. 
Clearly, such a scheme is feasible iff each of the allocation profiles $y(l), \forall l \in [k]$ belongs to $\pcal{F}$. 
Player $i$ thus faces the prospect $L_i = \{(1/k, y_i(l))\}_{l = 1}^k$
and such a lottery scheme is completely characterized by the tuple $y := (y_i(l), i \in [n], l \in [k]).$
By taking $k$ large enough, any lottery scheme can be approximated by such a scheme.

Let $y_i := (y_i(l))_{l \in [k]} \in \bbR^k_+$.
Let $z_i := (z_i(l))_{l \in [k]} \in \bbR^k_+$ be a vector and $\pi_i: [k] \to [k]$ be a permutation such that
\[
	z_i(1) \geq z_i(2) \geq  \dots \geq z_i(k),
\]
and
\[
	y_i(l) = z_i(\pi_i(l)) \mbox{ for all } l \in [k].
\]
Note that $y_i$ is completely characterized by $\pi_i$ and $z_i$.
Then player $i$'s CPT value will be
\[
	V_i(L_i) = \sum_{l = 1}^k h_i(l) v_i(z_i(l)),
\]
where $h_i(l) := w_i(l/k) -  w_i((l-1)/k)$ for $l \in [k]$. 
Let $h_i := (h_i(l))_{l \in [k]} \in \bbR^k_+$.
Note that $h_i(l) > 0$ for all $i,l$, since the weighting functions are assumed to be strictly increasing. 

Looking at the lottery scheme $y$ in terms of individual allocation profiles $z_i$ and permutations $\pi_i$ for all players $i \in [n]$, allows us to separate those features of $y$ that affect individual preferences and those that pertain to the network implementation. 
We will later see that the problem of optimizing aggregate utility can be decomposed into two layers:
\begin{inparaenum}[(i)]
	\item a convex problem that optimizes over resource allocations, and
	\item a non-convex problem that finds the optimal permutation profile.
\end{inparaenum}

Let $z := (z_i(l),i \in [n], l \in [k]), \pi := (\pi_i, i \in [n]), h := (h_i(l),i \in [n],l \in [k])$ and $v := (v_i(\cdot), i \in [n])$.
Let $S_{k}$ denote the set of all permutations of $[k]$.
The problem of optimizing aggregate utility 
$\sum_i V_i(L_i)$ subject to the lottery scheme being feasible,  
can be formulated as follows:
\begin{equation*}
\begin{aligned}
& \text{SYS}[z,\pi ; h,v,A,c]\\
& \text{Maximize}
& & \sum_{i=1}^n \sum_{l = 1}^k h_i(l) v_i(z_i(l))\\
& \text{subject to} && \sum_{i \in R_j} z_i(\pi_i(l)) \leq c_j, \forall j \in [m], \forall l \in [k],\\
&&& z_i(l) \geq z_i(l+1), \forall i \in [n], \forall l \in [k],\\
&&& \pi_i \in S_k, \forall i \in [n].
\end{aligned}
\end{equation*}
Here $R_j := \{i \in [n] | j \in J_i\}$ is the set of all players whose route uses link $j$. 
We set $z_i(k + 1) = 0$ for all $i$, and the $z_i(k+1)$ are not treated as variables.
This takes care of the condition $z_i(l) \geq 0$ for all $i \in [n], l \in [k]$.

%% file: tex/equilibrium_decomp.tex

\section{Equilibrium}
\label{sec: equilibrium}
The system problem $\SYS[z,\pi;h,v,A,c]$ optimizes over $z$ and $\pi$.
In this section we fix $\pi_i \in S_k$ for all $i$ and optimize over $z$.
Let us denote this fixed-permutation system problem by $\SYSRES[z;\pi,h,v,A,c]$.
\begin{equation*}
\begin{aligned}
& \SYSRES[z; \pi,h,v,A,c]\\
& \text{Maximize}
& & \sum_{i=1}^n \sum_{l = 1}^k h_i(l) v_i(z_i(l))\\
& \text{subject to} && \sum_{i \in R_j} z_i(\pi_i(l)) \leq c_j, \forall j \in [m], \forall l \in [k],\\
&&& z_i(l) \geq z_i(l+1), \forall i \in [n], \forall l \in [k].
\end{aligned}
\end{equation*}
(In contrast with $\SYS(z,\pi; \dots)$, in $\SYSRES(z; \pi, \dots)$, the permutation $\pi$ is thought of as being fixed.)
Since $v_i(\cdot)$ is assumed to be a concave function and $h_i(l) > 0$ for all $i,l$, this problem has a concave objective function with linear constraints. 
For all $j \in [m], l \in [k]$, let $\lambda_j(l) \geq 0$ be the dual variables corresponding to the constraints $\sum_{i \in R_j} z_i(\pi_i(l)) \leq c_j$ respectively,
and for all $i \in [n], l \in [k]$,
let $\alpha_i(l) \geq 0$ be the dual variables corresponding to the constraints $z_i(l) \geq z_i(l+1)$ respectively. 
Let $\lambda := (\lambda_j(l),j \in [m], l \in [k])$ and $\alpha := (\alpha_i(t), i \in [n], l \in [k])$.
Then the Lagrangian for the fixed-permutation system problem $\SYSRES[z;\pi,h,v,A,c]$ can be written as follows:
\begin{align*}
	\pcal{L}\l( z;\alpha,\lambda\r) &:= \sum_{i=1}^n \sum_{l=1}^k h_i(l)v_i(z_i(l))\\
	& + \sum_{i=1}^n \sum_{l=1}^k \alpha_i(l)[z_i(l) - z_i(l+1)] + \sum_{j = 1}^m\sum_{l=1}^k \lambda_j(l) [c_j - \sum_{i \in R_j} z_i(\pi_i(l))]\\
	&= \sum_{i=1}^n \sum_{l=1}^k \l[h_i(l)v_i(z_i(l)) + (\alpha_i(l) - \alpha_i(l-1))z_i(l) - \l(\sum_{j \in {J}_i} \lambda_j(\pi_i^{-1}(l))\r)z_i(l)\r]\\
	& + \sum_{j=1}^m\sum_{l=1}^k \lambda_j(l) c_j,
\end{align*}
where $\alpha_i(0) = 0$ for all $i \in [n]$. 
Differentiating the Lagrangian with respect to $z_i(l)$ we get,
\begin{equation*}
	\frac{\partial \pcal{L}(z;\alpha,\lambda)}{\partial z_i(l)} = h_i(l)v'_i(z_i(l)) + \alpha_i(l) - \alpha_i(l-1) -  \l(\sum_{j \in {J}_i} \lambda_j(\pi_i^{-1}(l))\r).
\end{equation*}
Let
\begin{equation}
\label{eq: rho_def}
	\rho_i(l) := \sum_{j \in {J}_i} \lambda_j(\pi_i^{-1}(l)),
\end{equation}
for all $i \in [n], l \in [k]$.
This can be interpreted as the price per unit throughput for player $i$ for her $l$-th largest allocation $z_i(l)$. 
The price of the lottery $z_i$ for player $i$ is given by
$\sum_{l=1}^k \rho_i(l) z_i(l)$, or equivalently,
\[
	\sum_{l=1}^k r_i(l)\l[z_i(l) - z_i(l+1)\r],
\]
where
\begin{equation}
\label{eq: def_rate_vec}
	r_i(l) := \sum_{s = 1}^l \rho_i(s), \text{ for all } l \in [k].
\end{equation}

For $l \in [k-1]$, $\alpha_i(l)$ can be interpreted as a transfer of a nonnegative price for player $i$ from her $l$-th largest allocation to her ($l+1$)-th largest allocation. 
Since the allocation $z_i(l+1)$ cannot be greater than the allocation $z_i(l)$, there is a subsidy of $\alpha_i(l)$ in the price of $z_i(l)$ and an equal surcharge of $\alpha_i(l)$ in the price of $z_i(l+1)$. 
This subsidy and surcharge is nonzero (and hence positive) only if the constraint is binding, i.e. $z_i(l) = z_i(l+1)$.
On the other hand, $\alpha_i(k)$ is a subsidy in price given to player $i$ for her lowest allocation, since she cannot be charged anything higher than the marginal utility at her zero allocation.

Let $h_i := (h_i(l))_{l \in [k]} \in \bbR^k_+$.
Consider the following user problem for player~$i$:
\begin{equation}
\begin{aligned}
& \text{USER}[m_i ; r_i,h_i,v_i]\\
& \text{Maximize}
& &\sum_{l = 1}^k h_i(l) v_i\l(\sum_{s=l}^k \frac{m_i(s)}{r_i(s)}\r) - \sum_{l = 1}^k m_i(l)\\
& \text{subject to} && m_i(l) \geq 0, \forall l \in [k],\\
\end{aligned}
\end{equation}
where $r_i := (r_i(l),l \in [k])$ is a vector of rates such that 
\begin{equation}
	\label{eq: rate_cond}
	0 < r_i(1) \leq r_i(2) \leq \dots \leq r_i(k).
\end{equation}
We can interpret this as follows: User $i$ is charged rate $r_i(k)$ for her lowest allocation $\delta_i(k) := z_i(k)$. 
Let $m_i(k)$ denote the budget spent on the lowest allocation and hence $m_i(k) = r_i(k)\delta_i(k)$. 
For $1 \leq l < k$, she is charged rate $r_i(l)$ for the additional allocation $\delta_i(l) := z_i(l) - z_i(l+1)$, beyond $z_i(l+1)$ up to the next lowest allocation $z_i(l)$. 
Let $m_i(l)$ denote the budget spent on $l$-th additional allocation and hence $m_i(l) = r_i(l)\delta_i(l)$. 

Let $m := (m_i(l), i \in [n], l \in [k])$ and $\delta := (\delta_i(l), i \in [n], l \in [k])$.
Consider the following network problem:
\begin{equation*}
\begin{aligned}
& \text{NET}[\delta;m,\pi,A,c]\\
& \text{Maximize}
& & \sum_{i=1}^n \sum_{l = 1}^k m_i(l) \log(\delta_i(l))\\
& \text{subject to} && \delta_i(l) \geq 0, \forall i, \forall l,\\
&&& \sum_{i \in R_j} \sum_{s = \pi_i(l)}^{k} \delta_i(s) \leq c_j, \forall j, \forall l.
\end{aligned}
\end{equation*}

This is the well known Eisenberg-Gale convex program \cite{eisenberg1959consensus}  and it can be solved efficiently. 
Kelly et al. \cite{kelly1998rate} proposed continuous time algorithms for finding equilibrium prices and allocations. For results on polynomial time algorithms for these problems see \cite{jain2010eisenberg,chakrabarty2006new}. 
We have the following decomposition result:

\begin{theorem}
\label{thm: decompostion}
	For any fixed $\pi$, there exist equilibrium parameters $r^*, m^*, \delta^*$ and $z^*$ such that
	\begin{enumerate}[(i)]
		\item for each player $i$, $m_i^*$ solves the user problem $\text{USER}[m_i;r_i^*,h_i,v_i]$,
		\item $\delta^*$ solves the network problem $\text{NET}[\delta ; m^*, \pi ,A,c]$,
		\item $m_i^*(l) = \delta_i^*(l) r_i^*(l)$ for all $i,l$,
		\item $\delta_i^*(l) = z_i^*(l) - z_i^*(l+1)$ for all $i,l$, and
		\item $z^*$ solves the fixed-permutation system problem $\SYSRES[z;\pi,h,v,A,c]$.
	\end{enumerate}
\end{theorem}
\begin{proof}
	Since $\SYSRES[z;\pi,h,v,A,c]$ is a convex optimization problem, 
we know that there exist $z^* = (z_i^*(l), i \in [n],l \in [k]), \alpha^* = (\alpha_i^*(l), i \in [n],l \in [k])$ and $\lambda^* = (\lambda_j^*(l), j \in [m],l \in [k])$ 
such that
\begin{align}
	h_i(l)v'_i(z_i^*(l)) = \rho_i^*(l) - \alpha_i^*(l) + \alpha_i^*(l-1), &\forall i, \forall l,\label{eq: sys_compl_1}\\ 
	z_i^*(l) \geq z_i^*(l+1), \quad \alpha_i^*(l) \geq 0, \quad \alpha_i^*(l) (z_i^*(l) - z_i^*(l+1)) = 0, &\forall i, \forall l,\label{eq: sys_compl_2}\\ 
	\sum_{i \in R_j} z_i^*(\pi_i(l)) \leq c_j, \quad \lambda_j^*(l) \geq 0, \quad \lambda_j^*(l)[c_j - \sum_{i \in R_j} z_i^*(\pi_i(l))] = 0, &\forall j, \forall l, \label{eq: sys_compl_3}
\end{align}
where
\begin{equation*}
	\rho_i^*(l) := \sum_{j \in {J}_i} \lambda_j^*(\pi_i^{-1}(l)),
\end{equation*}
and such that $z^*$ solves the fixed-permutation system problem $\SYSRES[z;\pi,h,v,A,c]$. Hence statement (v) holds for this choice of $z^*$.
Let $r_i^*(l) := \sum_{s = 1}^l \rho_i^*(s)$, $\delta_i^*(l) := z_i^*(l) - z_i^*(l+1)$ and $m_i^*(l) := \delta_i^*(l) r_i^*(l)$ for all $i,l$. 
	From \eqref{eq: sys_compl_1}, we have $\rho_i^*(1) > 0$, because $v_i(\cdot)$ is strictly increasing, $h_i(1) > 0$ and $\alpha_i^*(0) = 0$. 
	Thus, the rate vector $r_i^*$ satisfies \eqref{eq: rate_cond} for all $i$. 
	Also note that, by construction, the vectors $r^*, \delta^*, z^*$ and $m^*$ satisfy statements (iii) and (iv) of the theorem. 

	We now show that statement (i) holds. 
	Fix a player $i$. 
	Observe that the user problem $\USER[m_i;r_i^*,h_i,v_i]$ has a concave objective function since $h_i(l) > 0, r_i^*(l) > 0$ and $v_i(\cdot)$ is concave. 
	Differentiating the objective function of the user problem $\USER[m_i;r_i^*,h_i,v_i]$ with respect to $m_i(l)$ at $m_i = m_i^*$, we get
	\[
		\sum_{s = 1}^l \frac{h_i(s) v'_i(z_i^*(s))}{r_i^*(l)} - 1.
	\]
	From \eqref{eq: sys_compl_1}, we have
	\begin{align*}
		\sum_{s = 1}^l h_i(s)v'_i(z_i^*(l)) &= \sum_{s = 1}^l \rho_i^*(s) - \alpha_i^*(l)\\
		&= r_i^*(l) - \alpha_i^*(l).
	\end{align*}
	Thus,
	\[
		\sum_{s = 1}^l \frac{h_i(s) v'_i(z_i^*(s))}{r_i^*(l)} - 1 = -\frac{\alpha_i^*(l)}{r_i^*(l)} \leq 0,
	\]
	where equality holds iff $\alpha_i^*(l) = 0$.
	If $m_i^*(l) > 0$, then $\delta_i^*(l) > 0$ and hence, by \eqref{eq: sys_compl_2}, $\alpha_i^*(l) = 0$. 
	Since this holds for all $l \in [k]$, 
	these are precisely the conditions necessary and sufficient for the optimality of $m_i^*$ in the problem of user $i$, it being a convex problem.

	We now show that statement (ii) holds. 
	The Lagrangian corresponding to the network problem $\NET[\delta; m^*,\pi,A,c]$ can be written as follows:
	\begin{align*}
		\pcal{L}(\delta;\mu) = \sum_i \sum_{l = 1}^k m_i^*(l) \log(\delta_i(l)) + \sum_j \sum_l \mu_j(l) \l[c_j - \sum_{i \in R_j} \sum_{s = \pi_i(l)}^k \delta_i(s)\r],
	\end{align*}
	where $\mu_j(l)$ is the dual variable corresponding to the link constraint and $\mu := (\mu_j(l),j\in [n],l \in [k])$. 
	Let $\mu_j(l) = \lambda_j^*(l)$. 
	If $m_i^*(l) > 0$, then $\delta_i^*(l) > 0$, and differentiating the Lagrangian with respect to $\delta_i(l)$ at $\delta_i^*(l)$, we get
	\begin{align*}
		\frac{\partial \pcal{L}(\delta;\lambda^*)}{\partial \delta_i(l)} \Bigg|_{\delta_i(l) = \delta_i^*(l)} 
		&= \frac{m_i^*(l)}{\delta_i^*(l)} - \sum_{j \in {J}_i} \sum_{s = 1}^l \lambda^*_j(\pi_i^{-1}(s))\\
		&= r_i^*(l) - \sum_{s = 1}^l \rho_i^*(s) = 0.
	\end{align*}
	If $m_i^*(l) = 0$, then 
	$\frac{\partial \pcal{L}(\delta;\lambda^*)}{\partial \delta_i(l)} \big|_{\delta_i(l) = \delta_i^*(l)} \leq 0$ since $\lambda^*_j(l) \geq 0$ for all $j,l$. 
	Further, from \eqref{eq: sys_compl_3}, we have $\lambda_j^*(l) \l[c_j - \sum_{i \in R_j} \sum_{s = \pi_i(l)}^k \delta_i^*(s)\r] = 0$ for all $j,l$. 
	Thus, $\delta^*$ solves the network problem $\NET[\delta;m^*,\pi,A,c]$.
\end{proof}

Thus the fixed-permutation system problem can be decomposed into user problems -- one for each player -- and a network problem, for any fixed permutation profile $\pi$. 
Similar to the framework in \cite{kelly1998rate}, we have an iterative process as follows: 
The network presents each user $i$ with a rate vector $r_i$. 
Each user solves the user problem $\USER[m_i;r_i,h_i,v_i]$, and submits their budget vector $m_i$, 
The network collects these budget vectors $(m_i)_{i \in [n]}$ and solves the network problem $\NET[\delta; m^*, \pi, A, c]$ to get the corresponding allocation $z$ (which can be computed from the incremental allocations $\delta$) and the dual variables $\lambda$.
The network then computes the rate vectors corresponding to each user from these dual variables as given by \eqref{eq: rho_def} and \eqref{eq: def_rate_vec} 
and presents it to the users as updated rates. 
Theorem \ref{thm: decompostion} shows that the fixed-permutation system problem of maximizing the aggregate utility is solved at the equilibrium of the above iterative process.
If the value functions $v_i(\cdot)$ are strictly concave, then one can show that the optimal lottery allocation $z^*$ for the fixed-permutation system problem is unique. 
However, the dual variables $\lambda$, and hence the rates $r_i, \forall i$, need not be unique. 
Nonetheless, if one uses the continuous-time algorithm proposed in \cite{kelly1998rate} to solve the network problem, then a similar analysis as in \cite{kelly1998rate}, based on Lyapunov stability, shows that the above iterative process converges to the equilibrium lottery allocation $z^*$.

One of the permutation profiles, say $\pi^*$, solves the system problem. 
In the next section, we explore this in more detail.
However, it is interesting to note that, 
 for any fixed permutation profile $\pi$, any  deterministic solution is a special case of the lottery scheme $y$ with permutation profile $\pi$.
Thus, it is guaranteed that the solution of the fixed-permutation system problem for any permutation profile $\pi$ is at least as good as any deterministic allocation.
Here is a simple example, where a lottery-based allocation leads to strict improvement over deterministic allocations.

\begin{example}
	Consider a network with $n$ players and a single link with capacity $c$. 
	Let $n = 10$ and $c = 10$.
	For all players $i$, we employ the value functions and weighting functions suggested by Kahneman and Tversky \cite{tversky1992advances}, given by 
	\[
		v_i(x_i) = x_i^{\beta_i}, \beta_i \in [0,1],
	\]
	and 
	\[
		w_i(p_i) = \frac{p_i^{\gamma_i}}{(p^{\gamma_i} + (1-p)^{\gamma_i})^{1/{\gamma_i}}}, \gamma_i \in (0,1],
	\]
	respectively.
	We take $\beta_i = 0.88$ and $\gamma_i = 0.61$ for all $i \in [n]$. These parameters were reported as the best fits to the empirical data in \cite{tversky1992advances}. 
	The probability weighting function is displayed in Figure \ref{fig: prob_weight}.

	\input{tex/plots_for_CPT}

	By symmetry and concavity of the value function $v_i(\cdot)$, the optimal deterministic allocation is given by allocating $c/n$ to each player $i$.
	The aggregate utility for this allocation is $n*v_1(c/n) = 10$.

	Now consider the following lottery allocation: Let $k = n = 10$. 
	Let $\pi_i(l) - 1 = l + i (\text{mod } k)$ for all $i \in [n]$ and $l \in [k]$.
	Let $x \in [c/n,c]$ and $z_i(1) = x$ for all $i \in [n]$ and $z_i(l) = (c-x)/(n-1)$ for all $i \in [n]$ and $l = 2,\dots,k$.
	Note that this is a feasible lottery allocation.
	Such a lottery scheme can be interpreted as follows: Select a ``winning'' player uniformly at random from all the players. 
	Allocate her a reward $x$ and equally distribute the remaining reward $c - x$ amongst the rest of the players.
	The ex ante aggregate utility is given by
	\[
		n*\l[w_1(1/n)v_1(x) + (1 - w_1(1/n))v_1((c-x)/(n-1)) \r].
	\]
	This function achieves its maximum equal to $14.1690$ at $x = 9.7871$.
	Thus, the above proposed lottery improves the aggregate utility over any deterministic allocation.
	The optimum lottery allocation is at least as good as $14.1690$.
\end{example}


%% file: tex/plots_for_CPT.tex

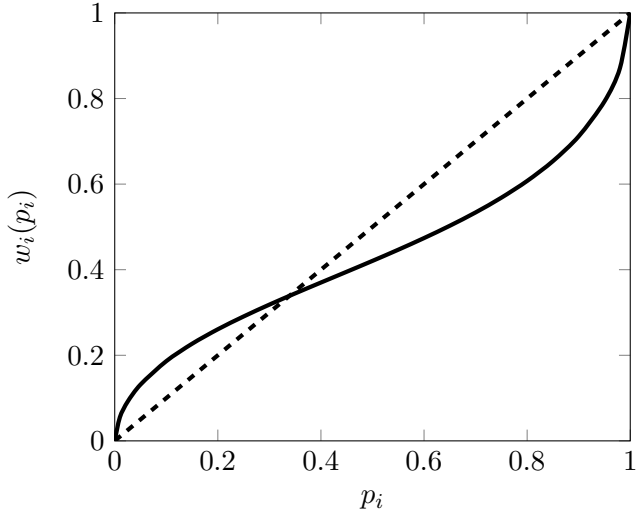
\begin{figure}
\centering

\begin{tikzpicture}[scale = 1]
    \begin{axis}[
    	xmin = 0,
    	xmax = 1,
    	ymin = 0,
    	ymax = 1,
        xlabel=$p_i$,
        ylabel=$w_i(p_i)$,
        every axis plot/.append style={ultra thick}]
    \addplot[smooth] plot coordinates {
        (0,0)
        (0.01, 0.0553)
        (0.025, 0.0915)
        (0.05, 0.1316)
        (0.10, 0.1863)
        (0.15, 0.2269)
        (0.20, 0.2608)
        (0.25, 0.2907)
        (0.30, 0.3184)
        (0.35, 0.3446)
        (0.40, 0.3700)
        (0.45, 0.3952)
        (0.50, 0.4206)
        (0.55, 0.4467)
        (0.60, 0.4739)
        (0.65, 0.5027)
        (0.70, 0.5338)
        (0.75, 0.5683)
        (0.80, 0.6074)
        (0.85, 0.6537)
        (0.90, 0.7117)
        (0.95, 0.7932)
        (0.98, 0.8710)
        (1, 1)
    };

    \addplot[smooth, dashed] plot coordinates {
        (0,0)
        (1, 1)
    };
    \end{axis}
\end{tikzpicture}
\caption{Probability weighting function}
\label{fig: prob_weight}
\end{figure}

%% file: tex/opt_pi.tex

\section{Optimum permutation profile and duality gap}
\label{sec: permutation_duality_gap}


The system problem $\text{SYS}[z,\pi;h,v,A,c]$ can equivalently be formulated as
\begin{equation}
\begin{aligned}
& \max_{\substack{
	\pi_i \in S_k \forall i,\\
	z : z_i(l) \geq z_i(l+1) \forall i,l
}
} 
\; \;
\min_{\substack{
	\lambda	 \geq 0
}
}
& & \sum_{i=1}^n \sum_{l=1}^k h_i(l)v_i(z_i(l))\\
&&& + \sum_{j=1}^m \sum_{l=1}^k \lambda_j(l)\l[c_j - \sum_{i \in R_j} z_i(\pi_i(l))\r].
\end{aligned}
\refstepcounter{opt}
\tag{\Roman{opt}}\label{opt: max_min}
\end{equation}
Let $W_{ps}$ denote the value of this problem. 
It is equal to the optimum value of the system problem $\text{SYS}[z,\pi;h,v,A,c]$.
By interchanging the $\max$ and $\min$, we obtain the following dual problem:
\begin{equation}
\begin{aligned}
& \min_{\substack{
	\lambda \geq 0
}
} 
\; \;
\max_{\substack{
	\pi_i \in S_k \forall i,\\
	z : z_i(l) \geq z_i(l+1) \forall i,l
}
}
& & \sum_{i=1}^n \sum_{l=1}^k h_i(l)v_i(z_i(l))\\
&&& + \sum_{j=1}^m \sum_{l=1}^k \lambda_j(l)\l[c_j - \sum_{i \in R_j} z_i(\pi_i(l))\r].
\end{aligned}
\refstepcounter{opt}
\tag{\Roman{opt}}\label{opt: min_max}
\end{equation}
Let $W_{ds}$ denote the value of this dual problem. 
By weak duality, we know that $W_{ps} \leq W_{ds}$. 
For a fixed $\lambda \geq 0$ and a fixed $z$ that satisfies $z_i(l) \geq z_i(l+1),\forall i,l$, the optimum permutation profile $\pi$ in the dual problem \eqref{opt: min_max}
should minimize
\[
	\sum_{j=1}^m \sum_{l=1}^k \lambda_j(l) \sum_{i \in R_j} z_i(\pi_i(l)),
\]
which equals
\[
	\sum_i \sum_l \hat \rho_i(l) z_i(\pi_i(l)),
\]
Here $\hat \rho_i(l) := \sum_{j \in {J}_i} \lambda_j(l)$, is the price per unit allocation for player $i$ under outcome $l$.
Since the numbers $z_i(l)$ are ordered in descending order, any optimal permutation $\pi_i$ must satisfy
\begin{equation}
\label{eq: pi_opt_order}
	\hat \rho_i(\pi_i^{-1}(1)) \leq \hat \rho_i(\pi_i^{-1}(2)) \leq \dots \leq \hat \rho_i(\pi_i^{-1}(k)).
\end{equation} 
In other words, any optimal permutation profile $\pi$ of the dual problem \eqref{opt: min_max} must allocate throughputs in the order opposite to that of the prices $\hat \rho_i(l)$.

\begin{lemma}
\label{lem: opt_pi_strong_dual}
If strong duality holds between the problems \eqref{opt: max_min} and \eqref{opt: min_max}, 
then any optimum permutation profile $\pi^*$ satisfies \eqref{eq: pi_opt_order} for all $i$.
 \end{lemma} 
We prove this lemma in Appendix \ref{sec: appendix}.
In general, there is a non-zero duality gap between the problems \eqref{opt: max_min} and \eqref{opt: min_max} (see Example \ref{ex: duality_gap} for such an example where the optimum permutation profile $\pi^*$ does not satisfy \eqref{eq: pi_opt_order}).

The permutation $\pi_i$ can be represented by a $k \times k$ permutation matrix $M_i$, where $M_i(s,t) = 1$ if $\pi_i(s) = t$ and $M_i(s,t) = 0$ otherwise, for $s,t \in [k]$. 
The network constraints $\sum_{i \in R_j} z_i(\pi_i(l)) \leq c_j, \forall l \in [k],$ can equivalently be written as $\sum_{i \in R_j} M_i z_i \leq c_j \1$, 
where $\1$ denotes a vector of appropriate size with all its elements equal to $1$, and the inequality is coordinatewise.
A possible relaxation of the system problem is to consider doubly stochastic matrices $M_i$ instead of restricting them to be permutation matrices. 
A matrix is said to be doubly stochastic if all its entries are nonnegative and each row and column sums up to $1$. 
A permutation matrix is hence a doubly stochastic matrix.
Let $\Omega_k$ denote the set of all doubly stochastic $k \times k$ matrices
and let $\Omega_k^*$ denote the set of all $k \times k$ permutation matrices.

Let $M = (M_i, i \in [n])$ denote a profile of doubly stochastic matrices.
The relaxed system problem can then be written as follows:
\begin{equation*}
\begin{aligned}
& \SYSR[z,M ; h,v,A,c]\\
& \text{Maximize}
& & \sum_{i=1}^n \sum_{l = 1}^k h_i(l) v_i(z_i(l))\\
& \text{subject to} && \sum_{i \in R_j} M_i z_i \leq c_j \1, \forall j,\\
&&& z_i(l) \geq z_i(l+1), \forall i, \forall l,\\
&&& M_i \in \Omega_k, \forall i.
\end{aligned}
\end{equation*}
Then the corresponding primal problem can be written as follows:
\begin{equation}
\begin{aligned}
& \max_{\substack{
	M_i \in \Omega_k \forall i,\\
	z : z_i(l) \geq z_i(l+1) \forall l,\forall i
}
} 
\; \;
\min_{\substack{
	\lambda_j \geq 0, \forall j
}
}
& & \sum_i \sum_l h_i(l)v_i(z_i(l))\\
&&& + \sum_j \lambda_j^T\l[c_j\1 - \sum_{i \in R_j} M_iz_i\r].
\end{aligned}
\refstepcounter{opt}
\tag{\Roman{opt}}\label{opt: max_min_relax}
\end{equation}
where $\lambda_j = (\lambda_j(l))_{l \in [k]} \in \bbR^k_+$.
Let $W_{pr}$ denote the value of this problem. 
Interchanging $\min$ and $\max$ we get the corresponding dual: 
\begin{equation}
\begin{aligned}
& \min_{\substack{
	\lambda_j \geq 0, \forall j
}
}
\; \;
\max_{\substack{
	M_i \in \Omega_k \forall i,\\
	z : z_i(l) \geq z_i(l+1) \forall l,\forall i
}
} 
& & \sum_i \sum_l h_i(l)v_i(z_i(l))\\
&&& + \sum_j \lambda_j^T\l[c_j\1 - \sum_{i \in R_j} M_iz_i\r].
\end{aligned}
\refstepcounter{opt}
\tag{\Roman{opt}}\label{opt: min_max_relax}
\end{equation}
Let $W_{dr}$ denote the value of this problem. 
If the link constraints in the relaxed system problem hold then
\begin{equation}
\label{eq: relax_to_average}
	\frac{1}{k}\sum_{i \in R_j} \sum_{l = 1}^k z_i(l) = \frac{1}{k}\sum_{i \in R_j} \1^T M_i z_i \leq \frac{1}{k}\1^T c_j \1 = c_j.
\end{equation}
This inequality essentially says that the link constraints should hold in expectation. 
Thus we have the following average system problem:
\begin{equation*}
\begin{aligned}
& \SYSA[z;h,v,A,c]\\
& \text{Maximize}
& & \sum_{i=1}^n \sum_{l = 1}^k h_i(l) v_i(z_i(l))\\
& \text{subject to} && \sum_{i \in R_j} \frac{1}{k}\sum_{l=1}^k z_i(l) \leq c_j, \forall j,\\
&&& z_i(l) \geq z_i(l+1), \forall i, \forall l,\\
\end{aligned}
\end{equation*}
with its corresponding primal problem:
\begin{equation}
\begin{aligned}
& \max_{\substack{
	z : z_i(l) \geq z_i(l+1) \forall l,\forall i
}
} 
\; \;
\min_{\substack{
	\bar \lambda_j \geq 0, \forall j
}
}
& & \sum_i \sum_l h_i(l)v_i(z_i(l))\\
&&& + \sum_j {\bar\lambda_j}\l[c_j - \sum_{i \in R_j} \frac{1}{k}\sum_{l=1}^k z_i(l)\r],
\end{aligned}
\refstepcounter{opt}
\tag{\Roman{opt}}\label{opt: max_min_avg}
\end{equation}
and the dual problem:
\begin{equation}
\begin{aligned}
& \min_{\substack{
	\bar \lambda_j \geq 0, \forall j
}
}
\; \;
\max_{\substack{
	z : z_i(l) \geq z_i(l+1) \forall l,\forall i
}
} 
& & \sum_i \sum_l h_i(l)v_i(z_i(l))\\
&&& + \sum_j {\bar\lambda_j}\l[c_j - \sum_{i \in R_j} \frac{1}{k}\sum_{l=1}^k z_i(l)\r],
\end{aligned}
\refstepcounter{opt}
\tag{\Roman{opt}}\label{opt: min_max_avg}
\end{equation}
where $\bar \lambda_j \in \bbR$ are the dual variables corresponding to the link constraints.
Let $W_{pa}$ and $W_{da}$ denote the values of these primal and dual problems respectively.

Then we have the following relation:
\begin{theorem}
\label{thm: duality_ineq}
For any system problem defined by $h,v,A$ and $c$, we have
	\[
		W_{ps}  \leq W_{pr} = W_{pa} = W_{da} = W_{dr} = W_{ds}.
	\]
\end{theorem}
\begin{proof}

	As observed earlier, if we replace the condition $M_i \in \Omega_k$ in the primal relaxed problem~\eqref{opt: max_min_relax} with the condition $M_i \in \Omega_k^*$, we get the primal system problem~\eqref{opt: max_min}. Since $\Omega_k^* \subset \Omega_k$, we have $W_{ps} \leq W_{pr}$.

	The constraint $\sum_{i \in R_j} \frac{1}{k}\sum_{l=1}^k z_i(l) \leq c_j$ can equivalently be written as
	\[
		\sum_{i \in R_j} \bar M_i z_i \leq c_j \1,
	\]
	for all $j$, where $\bar M_i$ is the matrix with all its entries equal to $1/k$.
	Thus, if we replace $M_i \in \Omega_k$ in the primal relaxed problem~\eqref{opt: max_min_relax} with $\bar M_i$, we get the primal average problem~\eqref{opt: max_min_avg}. 
	This implies that $W_{pa} \leq W_{pr}$.
	However, as observed earlier in Equation \eqref{eq: relax_to_average}, if for some fixed allocations $z$ the link constraints are satisfied with respect to any doubly stochastic matrix $M_i$, then they are also satisfied with respect to $\bar M_i$. 
	Thus the maximum of the average system problem $\SYSA[z;h,v,A,c]$ is at least as much as the maximum of the relaxed system problem $\SYSR[h,v,A,c]$. 
	Since the maximum of the relaxed system problem is equal to the value of its corresponding primal problem, we get $W_{pa} \geq W_{pr}$. 
	Thus, we have established that $W_{pr} = W_{pa}$.

	The average system problem has a concave objective function with linear constraints. Thus, strong duality holds, and we get $W_{pa} = W_{da}$.

	We now show that $W_{da} = W_{dr}$. 
    From $W_{pr} \le W_{dr}$ and $W_{pr} = W_{pa} = W_{da}$,
    we get $W_{da} \leq W_{dr}$.
	Suppose $\lambda_j = \bar \lambda_j \1$. 
	Then the objective function of the relaxed dual problem~\eqref{opt: min_max_relax},
	\begin{align*}
		 &\sum_{i=1}^n \sum_{l=1}^k h_i(l)v_i(z_i(l))  + \sum_{j=1}^m \lambda_j^T\l[c_j\1 - \sum_{i \in R_j} M_i z_i\r]\\
		 &= \sum_{i=1}^n \sum_{l=1}^k h_i(l)v_i(z_i(l)) + \sum_{j=1}^m \bar \lambda_j \l[c_j - \sum_{i \in R_j} \frac{1}{k}\sum_{l=1}^k z_i(l) \r],
	\end{align*}
	equals the objective function of the dual average problem~\eqref{opt: max_min_avg}. 
	This implies that $W_{dr} \leq W_{da}$. 
	This established that $W_{da} = W_{dr}$.

	Since any doubly stochastic matrix $M_i$ is a convex combination of permutation matrices by the Birkhoff-von Neumann theorem, 
	in the dual problem~\eqref{opt: min_max_relax}, 
	for any fixed $\lambda_j, z_i$, the optimum can be achieved by a permutation matrix. 
	This established that $W_{dr} = W_{ds}$. 

	This completes the proof.
\end{proof}

Thus the duality gap is a manifestation of the ``hard'' link constraints.
In the proof of the above theorem we saw that the relaxed problem is ``equivalent'' to the average problem and strong duality holds for this relaxation. 
We will later study the average problem in further detail (Section~\ref{sec: relaxtion}).


We observed earlier in Lemma \ref{lem: opt_pi_strong_dual} that if strong duality holds in the system problem, then the optimum permutation profile $\pi^*$ 
satisfies $\eqref{eq: pi_opt_order}$. 
Consider a simple example of two players sharing a single link. 
Suppose that, at the optimum, $\lambda(l)$ are the prices for $l \in [k]$ corresponding to this link under the different outcomes, and suppose not all of these are equal. 
Then the optimum permutation profile of the dual problem will align both players' allocations in the same order, i.e. the high allocations of player $1$ will be aligned with the high allocations of player $2$. 
However, we can directly see from the system problem that an optimum $\pi^*$ should align the two players' allocations in opposite order. 
The following example builds on this observation and shows that strong duality need not hold for the system problem.

\begin{example}
	\label{ex: duality_gap}
	Consider the following example with two players $\{1,2\}$ and a single link with capacity $2.9$. 
	Let $k=2$. 
	Let the corresponding CPT characteristics of the two players be as follows:
	\begin{align*}
		h_1(1) &= \frac{1}{3}, & h_1(2) &= \frac{2}{3},\\
		h_2(1) &= \frac{5}{6}, & h_2(2) &= \frac{1}{6},\\
		v_1(x) &= \log (x + 0.05) + 3, & v_2(x) &= \frac{2 \log(x + 0.05) + 3(x + 0.05)}{5} + 3. 
	\end{align*}
	For this problem, it is easy to see that $\pi_1 = (1, 2)$ and $\pi_2 = (2, 1)$ is an optimal permutation. 
	Solving the fixed-permutation system problem with respect to this permutation we get optimal value equal to $7.5621$. 
	The corresponding variable values are
	\begin{align*}
		z_1(1) &= y_1(1) = 1.95, & z_1(2) &= y_1(2) = 0.95,\\
		z_2(1) &= y_2(2) = 1.95, & z_2(2) &= y_2(1) = 0.95,
	\end{align*}
	and the dual variable values are
	\begin{align*}
		\lambda_1(1) &= \frac{1}{6}, & \lambda_1(2) = \frac{2}{3},
	\end{align*}
	and $\alpha_i(l) = 0$, for $i = 1,2, l = 1,2$.
	One can check that these satisfy the KKT conditions.

	Let us now evaluate the value of the dual problem \eqref{opt: min_max}. 
	By symmetry, we can assume without loss of generality that $\lambda_1(1) \leq \lambda_1(2)$. 
	As a result, optimal permutations for the dual problem are given by $\pi_1 = \pi_2 = (1, 2)$. 
	For fixed $\lambda_1(1)$ and $\lambda_1(2)$, we solve the following optimization problem:
\begin{equation*}
\begin{aligned}
&\max_{\substack{
	z_1(1) \geq z_1(2) \geq 0\\
	z_2(1) \geq z_2(2) \geq 0
}
}
& & \frac{1}{3}\log(z_1(1) + 0.05) + \frac{2}{3}\log(z_1(2) + 0.05)\\
&&& + \frac{5}{6}\l[\frac{2\log(z_2(1) + 0.05) + 3(z_2(1) + 0.05)}{5}\r]\\
&&&+ \frac{1}{6}\l[\frac{2\log(z_2(2) + 0.05) + 3(z_2(2) + 0.05)}{5}\r]\\
&&&  - \lambda_1(1)[z_1(1) + z_2(1)] - \lambda_1(2)[z_1(2) + z_2(2)] \\
&&&+ 2.9[\lambda_1(1) + \lambda_1(2)] + 6.
\end{aligned}
\refstepcounter{opt}
\tag{\Roman{opt}}\label{opt: max_dual_ex_opt}
\end{equation*}

If $\lambda_1(1) \leq 0.5$, then the value of the problem \eqref{opt: max_dual_ex_opt} is equal to $\infty$ (let $z_2(1) \to \infty$). 
If $\lambda_1(1) > 0.5$ (and hence $\lambda_1(2) > 0.5$ because $\lambda_1(2) \geq \lambda_1(2)$), 
then we observe that the effective domain of maximization in the problem \eqref{opt: max_dual_ex_opt} is compact and problem \eqref{opt: max_dual_ex_opt} has a finite value. 
Hence it is enough to consider $\lambda_1(1) > 0.5$.
At the optimum there exist $\alpha_1(1),\alpha_1(2), \alpha_2(1), \alpha_2(2) \geq 0$ such that
\begin{align*}
	\lambda_1(1) &= \frac{1}{3} \frac{1}{z_1(1) + 0.05} + \alpha_1(1), \\
	\lambda_1(2) &= \frac{2}{3} \frac{1}{z_1(2) + 0.05} - \alpha_1(1) + \alpha_1(2),\\
	\lambda_1(1) &= \frac{1}{3} \frac{1}{z_2(1) + 0.05} + \frac{1}{2} + \alpha_2(1),\\
	\lambda_1(2) &= \frac{1}{15} \frac{1}{z_2(2) + 0.05} + \frac{1}{10} -\alpha_2(1) + \alpha_2(2),
\end{align*}
and
\begin{align*}
	\alpha_1(1)[z_1(1) - z_1(2)] &= 0, & \alpha_1(2)z_1(2) &= 0,\\
	\alpha_2(1)[z_2(1) - z_2(2)] &= 0, & \alpha_2(2)z_2(2) &= 0.
\end{align*}

We now consider each of the sixteen ($4 \times 4$) cases based on whether the inequalities $z_i(l) \geq z_i(l+1)$ for $i=1,2$ and $l = 1,2$, hold strictly or not. 

\textbf{Case A1 ($z_1(1) = 0, z_1(2) = 0$)}. 
Then $\lambda_1(1) \geq 1/0.15$.

\textbf{Case B1 ($z_1(1) > 0, z_1(2) = 0$)}. 
Then $\lambda_1(1) < 1/0.15 , \lambda_1(2) \geq 2/0.15$, and 
\[
	\alpha_1(1) = 0, z_1(1) = \frac{1}{3\lambda_1(1)} - 0.05.
\]

\textbf{Case C1 ($z_1(1) = z_1(2) > 0$)}. 
Then $\lambda_1(2)/2 \leq \lambda_1(1) \leq \lambda_1(2), \lambda_1(1) + \lambda_1(2) < 1/0.05$, and
\begin{equation*}
\label{eq: z_1_opt_2}
	\alpha_1(1) = \frac{2\lambda_1(1) - \lambda_1(2)}{3}, \alpha_1(2) = 0, z_1(1) = z_1(2) = \frac{1}{\lambda_1(1) + \lambda_1(2)} - 0.05.
\end{equation*}

\textbf{Case D1 ($z_1(1) > z_1(2) > 0$)}. 
Then $\lambda_1(1) < \lambda_1(2)/2, 0 < \lambda_1(1) < 1/0.15, 0 < \lambda_1(2) < 2/0.15$, and
\begin{equation*}
\label{eq: z_1_opt_1}
	\alpha_1(1) = 0, \alpha_1(2) = 0, z_1(1) = \frac{1}{3\lambda_1(1)} - 0.05, z_1(2) = \frac{2}{3\lambda_1(2)} - 0.05.
\end{equation*}

\textbf{Case A2 ($z_2(1) = 0, z_2(2) = 0$)}.
Then $\lambda_1(1) \geq (1/0.15) + 0.5$.

\textbf{Case B2 ($z_2(1) > 0, z_2(2) = 0$)}. 
Then $0.5 < \lambda_1(1) < 2.15/0.3, \lambda_1(2) \geq (1/0.75) + 0.1$, and 
\[
	\alpha_2(1) = 0, z_2(1) = \frac{2}{6\lambda_1(1) - 3} - 0.05.
\]

\textbf{Case C2 ($z_2(1) = z_2(2) > 0$)}. 
This is not possible since $\lambda_1(1) \leq \lambda_1(2)$.

\textbf{Case D2 ($z_2(1) > z_2(2) > 0$)}. 
Then $0.5 < \lambda_1(1) < 2.15/0.3 , 0.1 < \lambda_1(2) < 2.15/1.5$, and
\begin{equation*}
\label{eq: z_2_opt}
	\alpha_2(1) = 0, \alpha_2(2) = 0, z_2(1) = \frac{2}{6\lambda_1(1) - 3} - 0.05, z_2(2) = \frac{2}{30\lambda_1(2) - 3} - 0.05.
\end{equation*}

If case A1 or case A2 holds, then $\lambda_1(1) \geq 1/0.15$. For any fixed $\lambda_1(1), \lambda_1(2)$, by choosing $z_i(l), i=1,2, l = 1,2$ small enough (respecting the conditions imposed by the corresponding cases), we get that the value of problem \eqref{opt: max_dual_ex_opt} is greater than or equal to $2.9*(1/0.15) = 19.3333 > 7.5621$. 
Similarly, if case B1 holds, then $\lambda_1(2) \geq 2/0.15$, and we get that the value of problem \eqref{opt: max_dual_ex_opt} is greater than or equal to $2.9*(2/0.15) = 38.6667 > 7.5621$.
For the remaining $4$ cases, substituting the corresponding expressions for $z_i(l), i = 1,2, l= 1,2$ in the objective function \eqref{opt: max_dual_ex_opt} 
and evaluating the optimum over feasible pairs $(\lambda_1(1), \lambda_1(2))$ for each pair of cases $\{\text{C1,D1}\} \times \{\text{B2,D2}\}$, the minimum is achieved for the case (C1,D2) and has value equal to $8.2757$. Numerical evaluation for each of these cases gives rise to the minimum values as shown in table \ref{fig: dualiy_gap_cases}.
\begin{figure}[t]
 \centering
  	\begin{tabular}{c | c | c |}
  	    \multicolumn{1}{c}{}& \multicolumn{1}{c}{B2} & \multicolumn{1}{c}{D2}\\
  	    \cline{2-3}
  	  C1 & $10.2284$ & $8.2757$\\
  	 \cline{2-3}
  	  D1 & $10.1814$ & $9.5006$\\
  	 \cline{2-3}
  	\end{tabular}
  	\caption{The numbers in the cells denote the optimum value of the objective function \eqref{opt: max_dual_ex_opt} in the corresponding cases.}\label{fig: dualiy_gap_cases}
  \end{figure} 
Thus the optimal dual value is $8.2757$ and this is strictly greater than the primal value. 
\end{example}

\begin{theorem}
\label{thm: nphardness}
	The primal problem \eqref{opt: max_min} is NP-hard.
\end{theorem}
\begin{proof}
	We describe a polynomial time procedure that reduces an instance of the  integer partition problem to a special case of the primal problem. 
	Given a set of positive integers $\{c_1,c_2,\dots,c_n\}$, the integer partition problem is to find a subset $S \subset [n]$, such that
	\[
		\sum_{i \in S} c_i = \sum_{i \notin S} c_i.
	\]
	If such a set $S$ exists, then we say that an integer partition exists.
	Consider a network with $n$ players and $n+1$ link constraints given by
	\[
		y_i \leq c_i, \forall i \in [n], \text{ and } \sum_{i=1}^n y_i \leq \frac{\sum_{i =1}^n c_i}{2}.
	\]
	It is easy to realize a network with these link constraints. 
	Let $k = 2$.
	Let the CPT characteristics of all the players be as follows:
	\begin{align*}
		h_i(1) = 1 - \epsilon, h_i(2) = \epsilon, v_i(x_i) = x_i, \forall i \in [n],
	\end{align*}
	where $\epsilon = 1/10$.
	Let $W_{ps}$ denote the optimal value of the system problem. 
	We show that $W_{ps} \geq T := (1-\epsilon)\sum_{i \in [n]} c_i$ if and only if an integer partition exists.
	Suppose an integer partition exists and is given by the set $S$, consider the allocation $\pi_i = [1,2]$ if $i \in S$ and $\pi_i = [2,1]$ otherwise, $z_i(1) = c_i, z_i(2) = 0$ for all $i \in [n]$. 
	The aggregate utility for this allocation is equal to $T$ and hence $W_{ps} \geq T$. 
	Suppose $W_{ps} \geq T$.
	Then there an allocation, say $z^*$ and $\pi^*$ with aggregate utility at least $T$.
	Since $k = 2$, $\pi^*$ actually defines a partition of $[n]$, given by $S = \{i \in [n] : \pi_i(1) = 1\}$. 
	We have, the aggregate utility
	\[
		W(1) + W(2) \geq T,
	\]
	where
	\begin{align*}
		W(1) &:= \sum_{i \in S} (1 - \epsilon) z_i(1) + \sum_{i \notin S} \epsilon z_i(2),\\
		W(2) &:=  \sum_{i \notin S} (1 - \epsilon) z_i(1) + \sum_{i \in S} \epsilon z_i(2).
	\end{align*}
	Hence at least one of $W(1)$ and $W(2$ is at least as big as $T/2$.
	Without loss of generality, let $W(1) \geq T/2$.
	Thus we have,
	\[
		\sum_{i \in S} z_i(1) + \frac{\epsilon}{1 - \epsilon}\sum_{i \notin S} z_i(2) \geq \frac{\sum_{i \in [n]} c_i}{2}.
	\]
	However, since $z^*$ is feasible, the link constraints give
	\[
		\sum_{i \in S} z_i(1) + \sum_{i \notin S} z_i(2) \leq \frac{\sum_{i \in [n]} c_i}{2}.
	\]
	Since $\epsilon < 1/2$, we should have $\sum_{i \notin S} z_i(2) = 0$ and $\sum_{i \in S} z_i(1) = (\sum_{i \in [n]} c_i)/2$, implying that $S$ forms an integer partition.
	This completes the proof.
\end{proof}

%% file: tex/relaxation.tex

\section{Average System Problem and Optimal Lottery Structure}
\label{sec: relaxtion}

Suppose it is enough to ensure that the link constraints are satisfied in expectation, as in the average system problem. 
Consider the function $V_i^\text{avg}(\bar z_i)$ on $\bbR_+$ given by the value of the following optimization problem:
\begin{equation*}
\begin{aligned}
& \text{Maximize}
& &\sum_{l = 1}^k h_i(l) v_i(z_i(l))\\
& \text{subject to} && \frac{1}{k} \sum_{l = 1}^k z_i(l) = \bar z_i,\\
&&& z_i(l) \geq z_i(l+1), \forall l \in [k].\\
\end{aligned}
\refstepcounter{opt}
\tag{\Roman{opt}}\label{opt: avg_CPT_value}
\end{equation*}

Let $Z_i(\bar z_i)$ denote the set of feasible $(z_i(l))_{l \in [k]}$ in the above problem for any fixed $\bar z_i \geq 0$.
We observe that $Z_i(\bar z_i)$ is a closed and bounded polytope, and hence $V_i^\text{avg}(\bar z_i)$ is well defined. 

\begin{lemma}
\label{lem: envelope_diff}
	For any continuous, differentiable, concave and strictly increasing value function $v_i(\cdot)$, the function $V_i^\text{avg}(\cdot)$ is continuous, differentiable, concave and strictly increasing in $\bar z_i$.
\end{lemma}

We prove this lemma in Appendic \ref{sec: appendix}.
The average system problem $\SYSA[z;h,v,A,c]$ can be written as
\begin{equation*}
\begin{aligned}
& \text{Maximize}
& & \sum_{i=1}^n V_i^\text{avg}(\bar z_i)\\
& \text{subject to} && \sum_{i \in R_j} \bar z_i \leq c_j, \forall j,\\
&&& \bar z_i \geq 0, \forall i.
\end{aligned}
\end{equation*}
Kelly \cite{kelly1997charging} showed that this problem can be decomposed into user problems, one for each user $i$,
\begin{equation*}
\begin{aligned}
& \text{Maximize}
& &V_i^\text{avg}(\bar z_i) - \bar \rho_i \bar z_i\\
& \text{subject to} && \bar z_i \geq 0,\\
\end{aligned}
\end{equation*}
and a network problem,
\begin{equation*}
\begin{aligned}
& \text{Maximize}
& & \sum_{i=1}^n \bar \rho_i \bar z_i\\
& \text{subject to} && \sum_{i \in R_j} \bar z_i \leq c_j, \forall j,\\
&&& \bar z_i \geq 0, \forall i,\\
\end{aligned}
\end{equation*}
in the sense that there exist $\bar \rho_i \geq 0, \forall i \in [n]$, such that the optimum solutions $\bar z_i$ of the user problems, for each $i$, solve the network problem and the average system problem.
Note that this decomposition is different from the one presented in Section \ref{sec: equilibrium}. 
Here the network problem aims at maximizing its total revenue $\sum_{i=1}^n \bar \rho_i \bar z_i$, instead of maximizing a weighted aggregate utility where the utility is replaced with a proxy logarithmic function. 
The above decomposition is not as useful as the decomposition in Section \ref{sec: equilibrium} in order to develop iterative schemes that converge to equilibrium. 
However, the above decomposition motivates the following user problem:
\begin{equation*}
\begin{aligned}
& \text{USER\_AVG}[z_i; \bar \rho_i, h_i,v_i]\\
& \text{Maximize}
& &\sum_{l = 1}^k h_i(l) v_i(z_i(l)) - \frac{\bar \rho_i}{k} \sum_{l = 1}^k z_i(l)\\
& \text{subject to} && z_i(l) \geq z_i(l+1), \forall l \in [k],\\
\end{aligned}
\end{equation*}
where, as before, $z_i(k+1) = 0$.

We observed in Proposition \ref{thm: duality_ineq} that strong duality holds in the average system problem. 
Let $z^*$ be the optimum lottery scheme that solves this problem.
Then, first of all, $z^*$ satisfies $z_i^*(l) \geq z_i^*(l+1)\;\forall i,l$ and is feasible in expectation, i.e., $\bar z^* := (\bar z_i^*)_{i \in [n]} \in \pcal{F}$, where $\bar z_i^* := (1/k) \sum_{l} z_i^*(l)$. 
Further, $z^*$ optimizes the objective function of the average system problem.
Besides, there exist $\bar \lambda_j^* \geq 0$ for all $j$
 such that 
the primal average problem~\eqref{opt: max_min_avg} and the dual average problem~\eqref{opt: min_max_avg} each attain their optimum at $z^*, (\bar \lambda_j^* , j \in [m])$. 


For player $i$, consider the price $\bar \rho_i^* := \sum_{j \in J_i} \bar \lambda_j^*$, which is obtained by summing the prices $\bar \lambda_j^*$ corresponding to the links on player $i$'s route. 
From the dual average problem~\eqref{opt: min_max_avg}, fixing $\bar \lambda_j = \bar \lambda_j^* \; \forall j$, we get that the optimum lottery allocation $z_i^*$ for player $i$ should optimize the problem $\USERA[z_i;\bar \rho^*_i, h_i,v_i]$.

We now impose some additional conditions on the probability weighting function that are typically assumed based on empirical evidence and certain psychological arguments \cite{kahneman1979prospect}.
We assume that the probability weighting function $w_i(p_i)$ is concave for small values of the probability $p_i$ and convex for the rest. 
Formally, there exists a probability $\tilde p_i \in [0,1]$ such that $w_i(p_i)$ is concave over the interval $p_i \in [0,\tilde p_i]$ and convex over the interval $[\tilde p_i,1]$. 
Typically the point of inflection, $\tilde p_i$, is around $1/3$.

Let $w_i^* : [0,1] \to [0,1]$ be the minimum concave function that dominates $w_i(\cdot)$, i.e., $w_i^*(p_i) \geq w_i(p_i)$ for all $p_i \in [0,1]$. 
Let $p_i^* \in [0,1]$ be the smallest probability such that $w_i^*(p_i)$ is linear over the interval $[p_i^*,1]$. 

\begin{lemma}
\label{lem: w_concave_hull}
	Given the assumptions on $w_i(\cdot)$, we have $p_i^* \leq \tilde p_i$ and $w_i^*(p_i) = w_i(p_i)$ for $p_i \in [0,p_i^*]$. 
	If $p_i^* < 1$, then for any $p_i^1 \in [p_i^*,1)$, we have 
	\begin{equation}
		\label{eq: w_shape_prop}
		w_i(p_i) \leq w_i(p_i^1) + (p_i-p_i^1)\frac{1 - w_i(p_i^1)}{1 - p_i^1}.
	\end{equation}
	for all $p_i \in [p_i^1,1]$.
\end{lemma}

A proof of this lemma is included in Appendix \ref{sec: appendix}.
We now show that, under certain conditions, the optimal lottery allocation $z_i^*$ satisfies
\begin{equation}
\label{eq: opt_lottery_equal}
	z_i^*(l^*) = z_i^*(l^* + 1) = \dots = z_i^*(k),
\end{equation}
where $l^* := \min \{l \in [k] : (l-1)/k \geq p_i^*\}$, provided $p_i^* \leq (k-1)/k$. 
As a result, for a typical optimum lottery allocation, the lowest allocation occurs with a large probability approximately equal to $1-p_i^*$, and with a few higher allocations that we recognize as bonuses. 

\begin{proposition}
	\label{lem: opt_lottery}
	For any average user problem $\USERA[z_i;\bar \rho^*_i,h_i,v_i]$ with a strictly increasing, continuous,  differentiable and strictly concave value function $v_i(\cdot)$, and a strictly increasing continuous probability weighting function $w_i(\cdot)$ (satisfying $w_i(0) = 0$ and $w_i(1) = 1$) such that $p_i^* \leq (k-1)/k$, the optimum lottery allocation $z_i^*$ satisfies Equation \eqref{eq: opt_lottery_equal}.
\end{proposition}
\begin{proof}
	The Lagrangian for the average user problem $\USERA[z_i;\bar \rho^*_i,h_i,v_i]$ is
	\begin{align*}
		\pcal{L}(z_i;\alpha_i) &= \sum_{l=1}^k h_i(l) v_i(z_i(l)) - \frac{\bar \rho^*_i}{k} \sum_{l = 1}^k z_i(l) + \sum_{l=1}^k \alpha_i(l)[z_i(l) - z_i(l+1)],
	\end{align*}
    where $\alpha_i(l) \geq 0$ are the dual variables corresponding to the order constraints $z_i(l) \geq z_i(l+1)$, and $\alpha_i(0) = 0$.
	Differentiating with respect to $z_i(l)$, we get,
	\begin{align*}
		\frac{\partial \pcal{L}(z_i;\alpha_i)}{\partial z_i(l)} &= h_i(l) v_i'(z_i(l)) - \frac{\bar \rho^*_i}{k} +  \alpha_i(l) - \alpha_i(l-1).
	\end{align*}
	Since the problem $\USERA[z_i;\bar \rho^*_i,h_i,v_i]$ has a concave objective function and linear constraints, there exist $\alpha_i^*(l) \geq 0$ such that
		\begin{align}
		\label{eq: compl_avg_user_1}
			h_i(l) v_i'(z_i^*(l)) &=  \frac{\bar \rho^*_i}{k} - \alpha_i^*(l) + \alpha_i^*(l-1), \forall l \in [k],
		\end{align}
	and
		\begin{align}
		\label{eq: compl_avg_user_2}
			\alpha_i^*(l)[z_i^*(l) - z_i^*(l+1)] &= 0, \forall l \in [k].
		\end{align}
	If $z_i^*$ consists of identical allocations then it trivially satisfies Equation \eqref{eq: opt_lottery_equal}.
	If not, then there exists $l^1 \in \{2,\dots,k\}$ such that
	\[
		z_i^*(l^1 - 1) > z_i^*(l^1) = z_i^*(l^1 + 1) = \dots = z_i^*(k),
	\]
	i.e. $z_i^*(l^1)$ is the lowest allocation and occurs with probability $(k-l^1+1)/k$, and the next lowest allocation is equal to $z_i^*(l^1-1)$.
	Summing the equations corresponding to $l^1 \leq l \leq k$ from \eqref{eq: compl_avg_user_1}, we get
	\begin{align}
	\label{eq: l1_sum_1}
		\l[\sum_{s = l^1}^k h_i(s) \r] v_i'(z_i^*(l^1)) = \l(\frac{k-l^1+1}{k}\r) \bar \rho^*_i - \alpha_i^*(k) + \alpha_i^*(l^1-1).
	\end{align}
	The equation corresponding to $l = l^1 - 1$ in \eqref{eq: compl_avg_user_1} says,
	\begin{align}
	 	\label{eq: l1_sum_2}
	 	h_i(l^1-1)v_i'(z_i^*(l^1-1)) = \frac{\bar \rho^*_i}{k} - \alpha_i^*(l^1 - 1) + \alpha_i^*(l^1 - 2).
	 \end{align} 
	 Since $z_i^*(l^1 - 1) > z_i^*(l^1)$, from \eqref{eq: compl_avg_user_2}, we have $\alpha_i^*(l^1 -1) = 0$.
	 Thus from \eqref{eq: l1_sum_1} and \eqref{eq: l1_sum_2}, we have
	 \begin{align*}
	 	h_i(l^1-1) v_i'(z^*_i(l^1-1)) \geq \frac{\bar \rho^*_i}{k} \geq \frac{1}{k - l^1 +1} \l[\sum_{s = l^1}^k h_i(s)\r] v_i'(z^*_i(l^1)).
	 \end{align*}
	 Further, since $v_i(\cdot)$ is strictly concave and strictly increasing, $z_i^*(l^1-1) > z_i^*(l^1)$ implies $0 < v_i'(z_i^*(l^1-1) < v_i'(z_i^*(l^1))$. 
	 Thus,
	 \begin{equation}
	 \label{eq: contradiction_1}
	 	h_i(l^1 - 1) > \frac{1}{k - l^1 +1} \l[\sum_{s = l^1}^k h_i(s)\r].
	 \end{equation}
	 If $(l^1-2)/k \geq p_i^*$, 
	 then
	 \begin{align}
	 \label{eq: dec_weight_ineq}
	 	h_i(l^1-1) &= w_i\l(\frac{l^1-1}{k}\r) - w_i\l(\frac{l^1-2}{k}\r) \nonumber\\ 
	 	&\leq \frac{1}{k - l^1 +1} \l[ w_i(1) - w_i\l(\frac{l^1-1}{k}\r) \r] \nonumber\\ 
	 	&= \frac{1}{k - l^1 +1} \l[\sum_{s = l^1}^k h_i(s)\r].
	 \end{align}
	 where the inequality follows from \eqref{eq: w_shape_prop} with $p_i^1 = (l^1 -2)/k$ and $p_i = (l^1 - 1)/k$.
	However, \eqref{eq: dec_weight_ineq} contradicts \eqref{eq: contradiction_1} and hence $(l^1 - 2)/k < p_i^*$. 
	This proves the lemma.
\end{proof}

%% file: tex/conclusion.tex

\section{Conclusions and Future Work}
\label{sec: conclusion}

We saw that if we take the probabilistic sensitivity of players into account, 
then lottery allocation improves the ex ante aggregate utility of the players. 
We considered the RDU model, a special case of CPT utility, to model probabilistic sensitivity. 
This model, however, is restricted to reward allocations, and it would be interesting to extend it to a general CPT model with reference point and loss aversion. 
This will allow us to study loss allocations as in punishment or burden allocations, for example criminal justice, military drafting, etc.

For any fixed permutation profile, we showed the existence of equilibrium prices in a market-based mechanism to implement an optimal lottery. 
We also saw that finding the optimal permutation profile is an  NP-hard problem. 
We note that the system problem has parallels in cross-layer optimization in wireless \cite{lin2006tutorial} and multi-route networks \cite{wang2005cross}. 
Several heuristic methods have helped achieve approximately optimal solutions in cross-layer optimization. 
Similar methods need to be developed for our system problem. 
We leave this for future work.

The hardness in the system problem comes from hard link constraints. 
Hence, by relaxing these conditions to hold only in expectation, we derived some qualitative features of the optimal lottery structure
under the typical assumptions on the probability weighting function of each agent in the RDU model. 
As observed, the players typically ensure their minimum allocation with high probability, and gamble for higher rewards with low probability.

%% file: tex/appendix.tex

\section{Proofs}
\label{sec: appendix}


\begin{proof}[Proof of Lemma \ref{lem: opt_pi_strong_dual}]
	Suppose problem \eqref{opt: max_min} and its dual \eqref{opt: min_max} have the same value.
	The value of \eqref{opt: max_min} is same as that of the system problem $\SYS[z,\pi; h,v,A,c]$. 
	Let us denote the objective function by
	\begin{align*}
		\Theta(\pi, z , \lambda) := 
 &\sum_{i=1}^n \sum_{l=1}^k h_i(l)v_i(z_i(l))\\
& + \sum_{j=1}^m \sum_{l=1}^k \lambda_j(l)\l[c_j - \sum_{i \in R_j} z_i(\pi_i(l))\r].
	\end{align*}
	Since $\pcal {F}$ is a polytope, for any fixed permutation profile $\pi$, the set of feasible $z$ is closed and bounded.
	The function $\Theta(\pi,z,\lambda)$ is continuous in $z$ and hence the fixed-permutation system problem $\SYSRES[z;\pi, h,v,A,c]$ has a bounded value.
	We also note that this value is non-negative.
	Since there are finitely many permutation profiles $\pi \in \prod_{i} S_k$, maximizing over these, we get that the system problem has a bounded non-negative value,
	say $W$, achieved say at $z^*$ and $\pi^*$.
	Thus
	\begin{equation}
	\label{eq: opt_value_W}
		\sum_{i=1}^n \sum_{l=1}^k h_i(l)v_i(z^*_i(l)) = W,
	\end{equation}
	and the lottery $z^*$ is feasible with respect to the permutation profile $\pi^*$, i.e. $z_i^*(l) \geq z_i^*(l+1)$ for all $i \in [n], l \in [k]$ and 
	\begin{equation}\label{eq: duality_feasibility}
		\sum_{i \in R_j} z^*_i(\pi^*_i(l)) \leq c_j \text{ for all } j \in [m], l \in [k].
	\end{equation}
	If this were not true, then the minimum of the objective function $\Theta(\pi^*,z^*,\lambda)$ with respect to $\lambda$ would be $-\infty$ and not $W \geq 0$.

	The value of the dual problem \eqref{opt: min_max} is equal to $W$. 
	Consider the function $\Theta_d : \bbR_+^{m \times k} \to \bbR$, given by maximizing over the objective function in problem \eqref{opt: min_max}, with respect to $\pi$ and $z$ for a fixed $\lambda \geq 0$,
	\begin{align*}
		\Theta_d(\lambda) := \max_{\substack{
	\pi_i \in S_k \forall i,\\
	z : z_i(l) \geq z_i(l+1) \forall i,l
}
}
 \Theta(\pi,z,\lambda).
	\end{align*}
	We note that the function $\Theta_d(\lambda)$ is lower semi-continuous, since the function $\Theta(\pi,z,\lambda)$ is continuous in $\lambda$. 
	Since $\l(v_i(\cdot), \forall i\r)$ are concave strictly increasing functions, there exists a sufficiently large finite $\lambda$ such that $0 \leq M := \Theta_d(\lambda) < \infty$.
	It follows that the minimum of $\Theta_d(\lambda)$ is achieved over the domain defined by $\lambda_j(l) \in [0,M/(\min_{j} c_j)]$ for all $j,l$.
	Since this is a bounded region and the function $\Theta_d(\lambda)$ is lower semi-continuous, there exists a $\lambda^*$ such that $\Theta_d(\lambda^*) = \min_{\lambda \geq 0} \Theta_d(\lambda) = W$.

	Since $\Theta_d(\lambda^*) = W$, we have $\Theta(\pi^*,z^*,\lambda^*) \leq W$.
	However, from \eqref{eq: opt_value_W}, \eqref{eq: duality_feasibility} and the fact that $\lambda_j^*(l) \geq 0$ for all $j,l$ we get $\Theta(\pi^*,z^*,\lambda^*) \geq W$.
	Hence $\Theta(\pi^*,z^*,\lambda^*) = W$.
	Thus the maximum in the definition of $\Theta_d(\lambda^*)$ is achieved at $z^*,\pi^*$. 
	This implies $\pi^*_i$ satisfies \eqref{eq: pi_opt_order} for all $i$.
\end{proof}

\begin{proof}[Proof of Lemma \ref{lem: envelope_diff}]
Let $\bar z_i \geq 0$ and $\tau > 0$. 
Let $z_i^* \in Z_i(\bar z_i)$ be such that $V_i^\text{avg}(\bar z_i) = \sum_{l=1}^k h_i(l) v_i(z_i^*(l))$. 
We have, $(z_i^*(l) + \tau)_{l \in [k]} \in Z_i(\bar z_i + \tau)$ and
\begin{align*}
	V_i^\text{avg}(\bar z_i + \tau) \geq \sum_{l =1}^k h_i(l) v_i(z_i^*(l) + \tau) > \sum_{l =1}^k h_i(l) v_i(z_i^*(l)) = V_i^\text{avg}(\bar z_i),
\end{align*}
where the strict inequality follows from the fact that $v_i(\cdot)$ is strictly increasing. 
This establishes that the function $V_i^\text{avg}(\bar z_i)$ is strictly increasing.

Let $\bar z_i^1, \bar z_i^2 \geq 0$ and $\sigma \in [0,1]$. 
Let $z_i^1 \in Z_i(\bar z_i^1)$ and $z_i^2 \in Z_i(\bar z_i^2)$ be such that $V_i^\text{avg}(\bar z_i^1) = \sum_{l=1}^k h_i(l) v_i(z_i^1(l))$ and $V_i^\text{avg}(\bar z_i^2) = \sum_{l=1}^k h_i(l) v_i(z_i^2(l))$.
Let $z_i^\sigma := \sigma z_i^1(l) + (1-\sigma) z_i^2(l)$ and $\bar z_i^\sigma := \sigma \bar z_i^1(l) + (1-\sigma) \bar z_i^2(l)$. Then $z_i^\sigma \in Z_i(\bar z_i^\sigma)$ and by the concavity of $v_i(\cdot)$, we have
\begin{align*}
	V_i^\text{avg}(\bar z_i^\sigma) \geq \sum_{l=1}^k h_i(l) v_i(z_i^\sigma(l)) &\geq \sum_{l=1}^k h_i(l) \l[\sigma v_i(z_i^1(l)) + (1-\sigma) v_i(z_i^2(l))\r]\\
	&= \sigma V_i^\text{avg}(\bar z_i^1) + (1-\sigma) V_i^\text{avg}(\bar z_i^2).
\end{align*}
	This establishes that the function $V_i^\text{avg}(\bar z_i)$ is concave.
	This implies that $V_i^\text{avg}(\bar z_i)$ is continuous and directionally differentiable at each $\bar z_i > 0$, 
	and we have the following relation between its left and right directional derivatives (see, for example, \cite{rockafellar2015convex}):
	\begin{equation}
	\label{eq: direct_diff_ineq}
		\frac{d}{d\bar z_i} V_i^\text{avg}(\bar z_i - ) \geq \frac{d}{d\bar z_i} V_i^\text{avg}(\bar z_i + ), \text{ for all } \bar z_i > 0.
	\end{equation}
	Further, if $(\bar z_i^t)_{t \geq 1}$ is a sequence such that $\bar z_i^t \to 0$, 
	and $z_i^t \in Z_i(\bar z_i^t)$ for all $t \geq 1$, then $z_i^t(l) \to 0$ for all $l \in [k]$. 
	By the continuity of the function $v_i(\cdot)$, we have that the function $V_i^\text{avg}(\bar z_i)$ is continuous at $\bar z_i = 0$.

	Let $\bar z_i > 0$, and let $\tau^t := 1/t$, for $t \geq 1$. 
	As before, let $z_i^* \in Z_i(\bar z_i)$ be such that $V_i^\text{avg}(\bar z_i) = \sum_{l=1}^k h_i(l) v_i(z_i^*(l))$. 
	Since $\bar z_i > 0$ and $(1/k) \sum_{l=1}^k z_i^*(l) = \bar z_i$, we have $z_i^*(l) > z_i^*(l+1)$ for at least one $l \in [k]$. 
	Let $\hat l \in [k]$ be the smallest such $l$.
	For $t \geq 1$, let $z_i^{t+}$ be given by
	\[
		z_i^{t+}(l) := \begin{cases}
			z_i^*(l) + \frac{k}{\hat l} \tau^t, &\text{ for } 1 \leq l \leq \hat l,\\
			z_i^*(l), &\text{ for } \hat l < l \leq k.
		\end{cases}
	\]
	Note that $z_i^{t+} \in Z_i(\bar z_i + \tau^t)$. 
	We have,
	\begin{align*}
		&V_i^\text{avg}(\bar z_i + \tau^t) - V_i^\text{avg}(\bar z_i) - \frac{k \tau^t}{\hat l}\sum_{l = 1}^{\hat l} h_i(l) v_i'(z_i^*(l)) \\
		& \geq \sum_{l=1}^k h_i(l) v_i(z_i^{t+}(l)) - \sum_{l=1}^k h_i(l) v_i(z_i^*(l)) - \frac{k \tau^t}{\hat l}\sum_{l = 1}^{\hat l} h_i(l) v_i'(z_i^*(l)) \\
		& = \sum_{l = 1}^{\hat l} h_i(l) \l[v_i(z_i^*(l) + {k \tau^t}/{\hat l}) - v_i(z_i^*(l)) - \frac{k \tau^t}{\hat l} v_i'(z_i^*(l))\r].
	\end{align*}
	Let $\gamma^* := \frac{k}{\hat l}\sum_{l = 1}^{\hat l} h_i(l) v_i'(z_i^*(l))$. 
	We have,
	\begin{equation}
	\label{eq: direct_diff_right}
		\frac{d}{d\bar z_i} V_i^\text{avg}(\bar z_i + ) \geq \liminf_{t \to \infty} \frac{V_i^\text{avg}(\bar z_i + \tau^t) - V_i^\text{avg}(\bar z_i)}{\tau^t} \geq \gamma^*.
	\end{equation}
	Similarly, for $t \geq \lceil k/z_i^*(\hat l) \rceil$ (here, $\lceil \cdot \rceil$ denotes the ceiling function), let $z_i^{t-}$ be given by
	\[
		z_i^{t-}(l) := \begin{cases}
			z_i^*(l) - \frac{k}{\hat l} \tau^t, &\text{ for } 1 \leq l \leq \hat l,\\
			z_i^*(l), &\text{ for } \hat l < l \leq k.
		\end{cases}
	\]
	We observe that $z_i^{t-} \in Z_i(\bar z_i - \tau^t)$, for all $t \geq \lceil k/z_i^*(\hat l) \rceil$, and 
	we have
	\begin{align*}
		&V_i^\text{avg}(\bar z_i) - V_i^\text{avg}(\bar z_i - \tau^t) - \frac{k \tau^t}{\hat l}\sum_{l = 1}^{\hat l} h_i(l) v_i'(z_i^*(l)) \\
		& \leq \sum_{l=1}^k h_i(l) v_i(z_i^*(l)) - \sum_{l=1}^k h_i(l) v_i(z_i^{t-}(l)) - \frac{k \tau^t}{\hat l}\sum_{l = 1}^{\hat l} h_i(l) v_i'(z_i^*(l)) \\
		& = \sum_{l = 1}^{\hat l} h_i(l) \l[v_i(z_i^*(l)) - v_i(z_i^*(l) - {k \tau^t}/{\hat l}) - \frac{k \tau^t}{\hat l} v_i'(z_i^*(l))\r].
	\end{align*}
	This implies,
	\begin{equation}
	\label{eq: direct_diff_left}
		\frac{d}{d\bar z_i} V_i^\text{avg}(\bar z_i - ) \leq \limsup_{t \to \infty} \frac{V_i^\text{avg}(\bar z_i) - V_i^\text{avg}(\bar z_i - \tau^t)}{\tau^t} \leq \gamma^*.
	\end{equation}
	From \eqref{eq: direct_diff_ineq}, \eqref{eq: direct_diff_right} and \eqref{eq: direct_diff_left}, we have
	\begin{equation}
		\frac{d}{d\bar z_i} V_i^\text{avg}(\bar z_i - ) = \frac{d}{d\bar z_i} V_i^\text{avg}(\bar z_i + ) = \gamma ^*.
	\end{equation}
	This establishes that the function $V_i^\text{avg}(\bar z_i)$ is differentiable and completes the proof.
\end{proof}


\begin{proof}[Proof of Lemma \ref{lem: w_concave_hull}]
Consider the function $\bar w_i : [0,1] \to [0,1]$, given by
\[
	\bar w_i(p_i) := \begin{cases}
		w_i^*(p_i) &\text{ for }  0\leq p_i < \tilde p_i,\\
		w_i^*(\tilde p_i) + (p_i - \tilde p_i) \frac{1 - w_i^*(\tilde p_i)}{1 - \tilde p_i} &\text{ for }  \tilde p_i \leq p_i  \leq 1.
	\end{cases} 
\]
Since $w_i^*$ is concave on $[0,1]$, one can verify that the function $\bar w_i$ is also concave on $[0,1]$.
Since $w_i^*$ dominates $w_i$, we have $w_i^*(\tilde p_i) \geq w_i(\tilde p_i)$.
Since the function $w_i(p_i)$ is convex on the interval $[\tilde p_i,1]$, 
we have $w_i(p_i) \leq \bar w_i(p_i)$,
for $p_i \in [\tilde p_i, 1]$.
However, since $w_i^*$ is the minimum concave function that dominates $w_i$, we get $\bar w_i = w_i^*$.
Thus, $w_i^*$ is linear over the interval $[\tilde p_i,1]$, and hence $p_i^* \leq \tilde p_i$.

Suppose $\tilde p_i = 1$. Then $w_i(\cdot)$ is a concave function on the unit interval $[0,1]$, and hence $w_i^*(p_i) = w_i(p_i)$, for $p_i \in [0,1] \supset [0,p_i^*]$.
Further, if $p_i^* < 1$, then $w_i(p_i)$ is linear over $[p_i^*,1]$,
 and inequality \eqref{eq: w_shape_prop} holds, in fact, with equality. 
This completes the proof of Lemma \ref{lem: w_concave_hull}, if $\tilde p_i = 1$.

For the rest of the proof we assume $\tilde p_i < 1$.
Define $g_i: [0,1) \to \bbR_+$ as
\[
	g_i(p_i) := \frac{1 - w_i(p_i)}{1 - p_i}.
\]
We now provide an alternate characterization of the function $w_i^*$ and the point $p_i^*$.
Let $\hat p_i \in [0,1]$ be given by
\[
	\hat p_i := \min \argmin_{p_i \in [0,\tilde p_i]} \l\{g(p_i)\r\}.
\]
The existence of $\hat p_i$ is guaranteed by the continuity of the function $g_i(p_i)$ on the compact interval $[0,\tilde p_i]$.
Let $\hat a_i := g_i(\hat p_i)$.
Since $w_i(p_i)$ is convex over the interval $[\tilde p_i, 1]$, the function $g_i(p_i)$ is nondecreasing over $[\tilde p_i,1)$.
Hence $g_i(p_i) \geq g_i(\hat p_i)$, for $p_i \in [0,1)$.
Substituting the expression for $g_i(\cdot)$ and rearranging, we get
\[
 	w_i(p_i) \leq w_i(\hat p_i) + \hat a_i (p_i - \hat p_i),
 \] 
 for $p_i \in [0,1]$.
 Since the function $w_i(p_i)$ is concave on the interval $[0,\hat p_i]$ and the linear function $w_i(\hat p_i) + \hat a_i (p_i - \hat p_i)$ dominates $w_i(p_i)$ on $[0,1]$, we have that the following function $\hat w_i(p_i)$ is concave on $[0,1]$:
 \[
 	\hat w_i(p_i) := \begin{cases}
		w_i(p_i) &\text{ for }  0 \leq p_i < \hat p_i,\\
		w_i(\hat p_i) + \hat a_i (p_i - \hat p_i) &\text{ for }  \hat p_i \leq p_i  \leq 1.
	\end{cases} 
 \]
 It follows that $w_i^*(p_i) = \hat w_i(p_i)$ for $p_i \in [0,1]$.
 Thus, $p_i^* \leq \hat p_i$ and $w_i^*(p_i) = \hat w_i(p_i) = w_i(p_i)$ for $p_i \in [0,p_i^*]$.
 If $\hat p_i = 0$, then $p_i^* = \hat p_i$.
 If $\hat p_i > 1$, then from the definition of $\hat p_i$, we have $g_i(p_i) > g_i(\hat p_i)$ for $p_i \in [0,\hat p_i)$, and this implies that $\hat p_i = p_i^*$.

We now prove inequality \eqref{eq: w_shape_prop}. 
Since $\tilde p_i < 1$ we have $p_i^* < 1$.
Rearranging we get that inequality \eqref{eq: w_shape_prop} is equivalent to showing that the function $g_i(p_i)$ is non-decreasing over the interval $[p_i^*,1)$.
As observed earlier, $g_i(p_i)$ is non-decreasing over the interval $[\tilde p_i,1)$. Hence it is enough to show that the function $g_i(p_i)$ is non-decreasing on $[p_i^*,\tilde p_i]$.
Suppose, on the contrary, there exist $p_i^1, p_i^2 \in [p_i^*,\tilde p_i]$ such that $p_i^1 < p_i^2$ and $g_i(p_i^1) > g_i(p_i^2)$. 
Since $p_i^* = \hat p_i$, and from the definition of $\hat p_i$, we have $p_i^1 > p_i^*$ and $g_i(p_i^*) \leq g_i(p_i^2)$. 
Since, $g_i(p_i)$ is a continuous function, there exist $p_i \in [p_i^*, p_i^1)$ such that $g_i(p_i) = g_i(p_i^2)$.
Thus we have $p_i < p_i^1 < p_i^2$ such that $g_i(p_i^1) > g_i(p_i) = g_i(p_i^2)$.
However, this contradicts the concavity of $w_i(p_i)$ on $[p_i^*,\tilde p_i]$.
This completes the proof.
\end{proof}
